\documentclass[a4paper, 12pt, conference,onecolumn]{ieeeconf}      % Use this line for a4
\IEEEoverridecommandlockouts                              % This command is only
\title{\Large\bf A quantum mechanical version of Price's theorem
for Gaussian states}%and Cameron-Martin-Girsanov Change of Measure

\author{Igor G. Vladimirov%\\
\thanks{This work is supported by the Australian Research Council.
The author is with UNSW Canberra, ACT 2600, Australia. E-mail: {\tt\small  igor.g.vladimirov@gmail.com}.}
}

\usepackage{amsmath}
\usepackage{amssymb}
\usepackage{amsfonts}
\usepackage{graphicx}
\usepackage{color}
\usepackage{times}
\usepackage{mathptmx} % assumes new font selection scheme installed
\usepackage{datetime}
\usepackage{framed}

\usepackage{datetime}
\usepackage{times} % assumes new font selection scheme installed
\usepackage{framed} % for pdf, bitmapped graphics files
\usepackage{graphicx} % for pdf, bitmapped graphics files
\usepackage{amsmath} % assumes amsmath package installed
\usepackage{amssymb}  % assumes amsmath package installed
\newtheorem{thm}{Theorem}

\def\fa{\mathfrak{a}}

\def\<{\leqslant}           % nice less than or equal to sign
\def\>{\geqslant}           % nice larger than or equal to sign
\def\d{\partial}

\def\wt{\widetilde}

\def\Re{\mathrm{Re}}   % real part
\def\Im{\mathrm{Im}}   % imaginary part

\def\cH{\mathcal{H}}   % Hardy space
\def\mA{\mathbb{A}}    % space of real antisymmetric matrices
\def\mZ{\mathbb{Z}}    % set of integers
\def\mR{\mathbb{R}}    % real line
\def\mC{\mathbb{C}}    % complex plane

\def\Tr{\mathrm{Tr}}       % matrix trace
\def\rT{\mathrm{T}}        % matrix transpose
\def\bE{\mathbf{E}}    % expectation

\def\rprod{\mathop{\overrightarrow{\prod}}}

\def\[[[{[\![\![}
\def\]]]{]\!]\!]}

\def\bra{{\langle}}
\def\ket{{\rangle}}

\def\Bra{\left\langle}
\def\Ket{\right\rangle}

\def\re{\mathrm{e}}        % number e
\def\rd{\mathrm{d}}        % differential

\def\bJ{\mathbf{J}}

\def\x{\times}
\def\ox{\otimes}

\def\cI{\mathcal{I}}

\def\cov{\mathbf{cov}}

\def\mS{\mathbb{S}}

\def\mZ{\mathbb{Z}}

\def\veps{\varepsilon}

\begin{document}

\maketitle
\thispagestyle{empty}
\pagestyle{plain}

%\author{Igor G. Vladimirov}
%\address{UNSW Canberra, ACT 2600, Australia}
%%\curraddr{UNSW Canberra, ACT 2600, Australia}
%\email{igor.g.vladimirov@gmail.com}
%\thanks{This work is supported by the Australian Research Council}

%\subjclass[2000]{Primary }
%    For articles to be published after 1 January 2010, you may use
%    the following version:
%\subjclass[2010]{Primary }

%\keywords{Price's theorem, quantum variables, canonical commutation relations, Weyl quantization, Gaussian quantum state, generalized moment, integro-differential equation.%, risk-sensitive moment.
%}

%\date{ \today, \currenttime}

%\dedicatory{}

\begin{abstract}
This paper is concerned with integro-differential identities which are known in statistical signal processing as Price's theorem for expectations of nonlinear functions of jointly Gaussian random variables. We revisit these relations for classical variables by using the  Frechet differentiation with respect to covariance matrices, and then show that Price's theorem carries over to a quantum mechanical setting. The  quantum counterpart of the theorem is established for Gaussian quantum states in the framework of the Weyl functional calculus for quantum variables satisfying the Heisenberg canonical commutation relations. The quantum mechanical version of Price's theorem relates the Frechet derivative of the generalized moment of such variables with respect to the real part of their quantum covariance matrix  with other moments. As an illustrative example, we consider these relations for quadratic-exponential moments which are relevant to risk-sensitive quantum control.
%, and discuss their links to an analogue of the Cameron-Martin-Girsanov change of measure for Gaussian quantum states.
\end{abstract}
\begin{keywords}
Price's theorem, quantum variables, canonical commutation relations, Weyl quantization, Gaussian quantum state, generalized moment,
integro-differential identity, quadratic-exponential moment. %risk-sensitive quantum control.%, Cameron-Martin-Girsanov theorem.
\end{keywords}

%%%%%%%%%%%%%%%%%%%%%%%%%%%%%%%%%%%%%%%%%%%%%%%%%%%%%%%%%%%%%%%%%%%%%%%%%%%%%%%%%%%%%%%%%%%%%%%%%%%
\section{INTRODUCTION}\label{sec:intro}
%%%%%%%%%%%%%%%%%%%%%%%%%%%%%%%%%%%%%%%%%%%%%%%%%%%%%%%%%%%%%%%%%%%%%%%%%%%%%%%%%%%%%%%%%%%%%%%%%%%

It is well known that Gaussian probability density functions (PDFs) provide fundamental solutions of the heat (or diffusion) equation for homogeneous media \cite{E_1998,V_1971} or the more general Fokker-Planck-Kolmogorov equation \cite{S_2008} for linear stochastic systems driven by a Wiener process. This connection between Gaussian PDFs and  linear second-order partial differential equations (PDEs) %with constant coefficients
is a source of various integro-differential identities.  In statistical signal processing,  such relations attracted attention more than fifty years ago
in the context of evaluating the generalized moments (that is, expectations of arbitrary nonlinear functions) of Gaussian random variables and became known under the generic name of Price's theorem \cite{P_1958}; see also \cite{B_1967,M_1964,P_1965,V_1999}.

The identities, which constitute  Price's theorem, relate the derivatives of the generalized moments with respect to the covariances of the Gaussian random variables with the expectations of the second-order derivatives of the nonlinear function. The latter moments can be easier to compute (for example, in the case of polynomials). Moreover, if the function satisfies a linear PDE with constant coefficients, then a ``dual'' PDE can be derived for the corresponding generalized moment. This, in principle,   allows the moment to be computed (or the structure of its parameter dependence to be found) by solving a boundary value problem for the dual PDE, where the boundary conditions are obtained by using extreme values of the parameters of the Gaussian distribution, for which the moment is amenable to direct calculation.% (for example, from symmetry considerations).

A similar problem of averaging nonlinear functions of quantum variables often arises in the context of quantum stochastic systems \cite{P_1992}.  These are models of open  dynamical systems with noncommutative variables, which evolve in time and interact  with their environment according to the laws of quantum mechanics \cite{S_1994}. For example, Gaussian stochastic linearization \cite{VP_2012a} of such systems, with dynamic variables satisfying   the Heisenberg   canonical commutation relations (CCRs) \cite{M_1998}, employs mixed moments of the system variables and the Hamiltonian operators   over Gaussian quantum states \cite{P_2010}. Furthermore, the performance criteria in quantum formulations of risk-sensitive dissipativity analysis and filtering/control design \cite{DDJW_2006,J_2005,VP_2012b,YB_2009} are organised as exponential moments of quadratic forms of quantum variables and are concerned with Gaussian quantum states in the case of linear systems \cite{P_2010}. % driven by external boson fields in vacuum states.

A straightforward extension of Price's theorem to the quantum mechanical setting is complicated by the nontrivial problem of evaluating a nonlinear function (of several classical variables) at noncommutative quantum variables. In fact, such evaluation can be carried out in different ways, leading to different results which depend on additional conventions on ordering of the variables in their products, such as Wick's normal order and  the related Kohn-Nirenberg calculus \cite{F_1989}.

In the present paper, we show that a quantum analogue of Price's theorem  can be established in the framework of the Weyl functional calculus \cite{F_1989}, which extends classical functions to quantum variables (satisfying the Heisenberg CCRs) by using the Fourier transforms. More precisely, the quantum mechanical derivation combines the Fourier transforms with the quantum quasi-characteristic functions \cite{CH_1971} of Gaussian quantum states. This is similar to  the role of characteristic functions in one of the  proofs of the classical version of Price's theorem in \cite{B_1967,V_1999}. The resulting quantum version of Price's theorem extends its particular cases which are known in quantum optics (where they are usually  formulated in terms of the annihilation and creation operators); see, for example, \cite{GZ_2004,RA_1978} and references therein.

The paper is organised as follows. In order to make the exposition self-contained,  Section~\ref{sec:classical} revisits Price's classical theorem by using the Frechet differentiation with respect to covariance matrices. Section~\ref{sec:quantum} outlines the Weyl quantization and establishes the quantum version of Price's theorem for Gaussian states. Section~\ref{sec:examples} provides an illustrative example which considers
%quantum Price's theorem for
the quadratic-exponential moments. % and discusses their links with the Cameron-Martin-Girsanov change of measure.
Section~\ref{sec:conc} makes concluding remarks.

%%%%%%%%%%%%%%%%%%%%%%%%%%%%%%%%%%%%%%%%%%%%%%%%%%%%%%%%%%%%%%%%%%%%%%%%%%%%%%%%%%%%%%%%%%%%%%%%%%%
\section{PRICE'S THEOREM FOR CLASSICAL GAUSSIAN RANDOM VARIABLES}\label{sec:classical}
%%%%%%%%%%%%%%%%%%%%%%%%%%%%%%%%%%%%%%%%%%%%%%%%%%%%%%%%%%%%%%%%%%%%%%%%%%%%%%%%%%%%%%%%%%%%%%%%%%%

Suppose $X:= (X_k)_{1\< k \< n}$ is an $\mR^n$-valued Gaussian random vector  with mean $\mu:= \bE X$ and a positive definite covariance matrix $\Sigma:= \cov(X) := \bE(XX^{\rT}) - \mu\mu^{\rT}$, where $\bE(\cdot)$ denotes expectation, and $(\cdot)^{\rT}$ is the transpose.   Unless indicated otherwise, vectors are organised as columns. Consider a generalized moment
\begin{equation}
\label{g}
    \bE f(X) = \int_{\mR^n} f(x)p_{\mu,\Sigma}(x)\rd x=: g(\mu,\Sigma)
\end{equation}
of the vector $X$,  which is specified by a function $f: \mR^n \to \mR$, where
\begin{equation}
\label{p}
    p_{\mu,\Sigma}(x):= \frac{(2\pi)^{-n/2}}{\sqrt{\det \Sigma}}\re^{-\frac{1}{2}\|x-\mu\|_{\Sigma^{-1}}^2},
    \quad
    x \in \mR^n,
\end{equation}
is the corresponding  Gaussian PDF.
Here, $\|v\|_M:= \sqrt{v^{\rT} M v}$ is the Euclidean (semi-)norm of a vector $v$ associated with a real positive (semi-)definite symmetric matrix $M$. Differentiation of this PDF with respect to $\mu$ at a given but otherwise arbitrary point $x$ leads to
\begin{equation}
\label{dp/dmu}
    \d_{\mu}\ln p_{\mu,\Sigma}
    =
    \frac{\d_{\mu}p_{\mu,\Sigma}}{p_{\mu,\Sigma}}
    =
     \Sigma^{-1}(x-\mu).
\end{equation}
Similarly, the logarithmic Frechet derivative of the PDF $p_{\mu,\Sigma}$ with respect to the covariance matrix $\Sigma$ as an element of the Hilbert space $\mS_n$ of real symmetric matrices of order $n$ (endowed with the Frobenius inner product \cite{HJ_2007} of such matrices $\bra K,N\ket := \Tr(KN)$) is
\begin{align}
\nonumber
    \d_{\Sigma}\ln p_{\mu,\Sigma}
    & =
    \frac{\d_{\Sigma}p_{\mu,\Sigma}}{p_{\mu,\Sigma}}\\
\label{dp/dSigma}
    &=\frac{1}{2}\left( \Sigma^{-1}(x-\mu)(x-\mu)^{\rT}\Sigma^{-1} - \Sigma^{-1}\right).
\end{align}
Here, use is also made of the following Frechet derivatives:
\begin{equation}
\label{F12}
    \d_{\Sigma} \ln\det \Sigma = \Sigma^{-1},
    \qquad
    \d_{\Sigma} (\|v\|_{\Sigma^{-1}}^2) = - \Sigma^{-1}vv^{\rT}\Sigma^{-1},
\end{equation}
where $v\in \mR^n$ is a constant vector.
The integro-differential identities, which are known under the generic name of Price's theorem, are based on the following relations between the derivatives in (\ref{dp/dmu}) and (\ref{dp/dSigma}) and the gradient vector  and Hessian matrix of the Gaussian PDF in (\ref{p}) with respect to the state variables:
\begin{equation}
\label{dp/dx}
    \d_{\mu}p_{\mu,\Sigma} = -\d_x p_{\mu,\Sigma},
    \qquad
    \d_{\Sigma}p_{\mu,\Sigma} = \frac{1}{2}\d_x^2 p_{\mu,\Sigma}.
\end{equation}
Note that the first of these equalities is valid for an arbitrary differentiable (not necessarily Gaussian)  PDF  $p(x-\mu)$ which involves $\mu$ as a shift parameter.
However, the second equality in (\ref{dp/dx}) is a manifestation of the role which \emph{Gaussian} PDFs play as fundamental solutions of the heat (or diffusion) PDEs  \cite{E_1998,V_1971} for homogeneous anisotropic media:
\begin{equation}
\label{heat}
    \d_t u
    =
    \frac{1}{2} \Bra K, \d_x^2 u\Ket,
\end{equation}
where $K\in \mS_n$ is a positive definite matrix of thermal conductivity (or diffusivity), and $u(t,x)$ describes the temperature (respectively,  concentration) at time $t\> 0$ and location $x\in \mR^n$. The fundamental solution of this PDE is provided by the Gaussian PDF $p_{0, Kt}$ (the corresponding  Gaussian distribution  converges weakly to the atomic probability measure concentrated at the origin of $\mR^n$ as $t\to 0+$). The solution of the initial value problem for the PDE (\ref{heat}) with a continuous initial condition $u(0,\cdot)$ (growing not too fast at infinity) is described by the convolution of the latter with the heat kernel $p_{0,Kt}$:
$$
    u(t,x) = \int_{\mR^n} p_{0,Kt}(x-y) u(0,y)\rd y.
$$
Indeed, by letting $\Sigma := Kt$ for all $t>0$ (with the time derivative $\dot{\Sigma} = K$), and combining the chain rule for composite functions with the second of the equalities (\ref{dp/dx}), it follows that
$$
    \d_t p_{0,\Sigma}
    =
    \Bra
        \dot{\Sigma},
        \d_{\Sigma} p_{0,\Sigma}
    \Ket
    =
    \frac{1}{2}
    \Bra
        K,
        \d_x^2 p_{0,\Sigma}
    \Ket.
$$
Now, suppose the function $f$ in (\ref{g}) is twice continuously differentiable and there exists $\veps >0$ such that the Hessian matrix of $f$ satisfies
\begin{equation}
\label{decay}
    p_{\mu,\Sigma}\d_x^2 f=o(\re^{-\veps |x|^2}),
    \qquad
    x\to \infty.
\end{equation}
This condition at infinity ensures the convergence and parametric differentiability for the following improper integrals:
% and makes  the flow of appropriate vector fields thro\-ugh boundaries of large enclosing domains in $\mR^n$  asymptotically vanish, thus allowing the integrals to be reduced via an improper version of Green's second identity with no boundary contributions.
\begin{align}
\nonumber
    \d_{\mu}g(\mu,\Sigma)
    & =
    \int_{\mR^n} f(x)\d_{\mu}p_{\mu,\Sigma}(x)\rd x\\
\nonumber
    & =
    -\int_{\mR^n} f(x)\d_xp_{\mu,\Sigma}(x)\rd x\\
\label{Price1}
    & =
    \int_{\mR^n} \d_x f(x)p_{\mu,\Sigma}(x)\rd x
    =
    \bE\d_x f(X),\\
\nonumber
    \d_{\Sigma}g(\mu,\Sigma)
    & =
    \int_{\mR^n} f(x)\d_{\Sigma}p_{\mu,\Sigma}(x)\rd x\\
\nonumber
    & =
    \frac{1}{2}
    \int_{\mR^n} f(x)\d_x^2p_{\mu,\Sigma}(x)\rd x\\
\label{Price2}
    & =
    \frac{1}{2}
    \int_{\mR^n} \d_x^2 f(x)p_{\mu,\Sigma}(x)\rd x
    =
    \frac{1}{2}
    \bE \d_x^2 f(X),
\end{align}
where use is made of (\ref{dp/dx}) and the integration by parts.
Since, as mentioned before, the first equality in (\ref{dp/dx}) holds for arbitrary differentiable PDFs with a shift parameter $\mu$,  the identity (\ref{Price1}) does not essentially employ the Gaussian nature of $p_{\mu,\Sigma}$. However, (\ref{Price2}) is indeed specific for Gaussian PDFs and
implies that  the partial derivatives of the generalized moment $g(\mu,\Sigma)$ in (\ref{g}) with respect to the  entries $\sigma_{jk}$ of the covariance  matrix $\Sigma:= (\sigma_{jk})_{1\<j,k\< n}$ satisfy the integro-differential relations
\begin{align}
\label{jj}
    \d_{\sigma_{jj}} g(\mu,\Sigma)
    &=     \frac{1}{2}
    \bE \d_{x_j}^2 f(X),\\
\label{jk}
    \d_{\sigma_{jk}} g(\mu,\Sigma)
    & =
    \bE \d_{x_j}\d_{x_k} f(X),
    \quad
    1\< j\ne k\< n,
\end{align}
which constitute Price's theorem.
Note that the $\frac{1}{2}$-factor is absent from (\ref{jk}) due to the symmetry of the covariance matrix $\Sigma$ and the Hessian matrix  $\d_{x}^2 f$. Indeed, for any given indices $j\ne k$, the first variation of $g(\mu,\Sigma)$ with respect to $\sigma_{jk}$ is
\begin{align*}
    \d_{\sigma_{jk}} g \delta\sigma_{jk}
    &=
    \Bra
        \d_{\Sigma} g,
        \delta \Sigma
    \Ket
     =
    \frac{1}{2}
    \Bra
    \bE\d_x^2 f(X),
    \delta\Sigma
    \Ket\\
    & =
    \frac{1}{2}
    \Bra
    \bE\d_x^2 f(X),
    e_je_k^{\rT} + e_ke_j^{\rT}
    \Ket\delta \sigma_{jk}\\
    & =
    \frac{1}{2}
    \left(
        e_k^{\rT}\bE\d_x^2 f(X)e_j
        +
        e_j^{\rT}\bE\d_x^2 f(X)e_k
    \right)\delta \sigma_{jk}\\
    & =
    \bE\d_{x_j}\d_{x_k} f(X) \delta \sigma_{jk},
\end{align*}
which implies (\ref{jk}). Here,
 $\delta\Sigma = \delta\sigma_{jk} (e_je_k^{\rT} + e_ke_j^{\rT})$ is the corresponding variation of $\Sigma$, and $e_k$ denotes the $k$th standard basis vector in $\mR^n$. Therefore, if the function $f$ is $2m$ times continuously differentiable and, together with its partial derivatives up to order $2m-1$, satisfies (\ref{decay}), then repeated differentiation of (\ref{jk}) leads to
\begin{align}
\label{jjl}
    \d_{\sigma_{jj}}^{\ell} g(\mu,\Sigma)
    &=
    2^{-\ell}
    \bE \d_{x_j}^{2\ell} f(X),\\
\label{jkl}
    \d_{\sigma_{jk}}^{\ell} g(\mu,\Sigma)
    & =
    \bE \d_{x_j}^{\ell}\d_{x_k}^{\ell} f(X),
    \quad
    1\< j\ne k\< n,
\end{align}
for all $\ell = 1, \ldots, m$;
cf. \cite[Eq. (3)]{P_1958} and \cite[Eq. (5)]{M_1964}.
The identities (\ref{Price1})--(\ref{jkl}) allow the moment $g$ of the Gaussian random vector $X$ in (\ref{g}) to be found by using the other moments,  associated with the derivatives of the function $f$, which can be easier to compute (for example, if $f$ is a polynomial). More generally, if the function $f$ satisfies a linear PDE with constant coefficients, then a ``dual'' PDE can be derived for the function $g$. This, in principle, allows $g$ to be found as a solution to the boundary value problem for the dual PDE, where the boundary conditions can be established by using those values of the parameters of the Gaussian distribution for which the moment lends itself to a straightforward calculation, for example, from symmetry considerations.  We will now demonstrate this technique  (in a similar fashion to the use of PDEs  in the proof of the main theorem in \cite{B_1967}) for the moment-generating function of the Gaussian distribution:
\begin{equation}
\label{momgen}
    g(\mu,\Sigma) := \bE f(X),
    \qquad
        f(x):= \re^{\lambda^{\rT}x},
    \qquad
    x \in \mR^n,
\end{equation}
where $\lambda \in \mR^n$ is fixed but otherwise arbitrary.  The function $f$ satisfies the following PDEs
\begin{equation}
\label{fPDE}
    \d_x f = f\lambda,
    \qquad
    \d_x^2 f= f\lambda \lambda^{\rT}.
\end{equation}
Therefore, by applying Price's theorem (\ref{Price1}) and (\ref{Price2}) to (\ref{momgen}), and using (\ref{fPDE}), it follows that the function $g$ satisfies the dual PDEs
\begin{align}
\label{gPDE1}
    \d_{\mu} g
    &= \bE\d_x f(X) = \lambda \bE f(X) = g\lambda,\\
\label{gPDE2}
    \d_{\Sigma} g
    &= \frac{1}{2}\bE\d_x^2 f(X) = \frac{1}{2}\lambda\lambda^{\rT} \bE f(X) = \frac{1}{2}g\lambda\lambda^{\rT}.
\end{align}
Since the moment $g$ in (\ref{momgen}) takes positive values, the PDEs (\ref{gPDE1}) and (\ref{gPDE2}) are representable in an equivalent logarithmic form:
$$    \d_{\mu} \ln g
     =
    \lambda,
    \qquad
    \d_{\Sigma} \ln g
    =
    \frac{1}{2}\lambda\lambda^{\rT}.
$$
The right-hand sides of these two PDEs are independent of $\mu$ and $\Sigma$, and hence, their general solution $\ln g$ is an affine function of $\mu$ and $\Sigma$:
\begin{equation}
\label{logg}
    \ln g  = C + \lambda^{\rT}\mu + \frac{1}{2}\Bra \lambda\lambda^{\rT}, \Sigma\Ket
     = C + \lambda^{\rT}\mu + \frac{1}{2}\|\lambda\|_{\Sigma}^2.
\end{equation}
The additive constant $C$, which depends on $\lambda$,  is calculated as $C=\ln g(0,0) = 0$ from the boundary condition $g(0,0)= \bE\re^{\lambda^{\rT}0}=1$. The latter follows from the fact that the Gaussian random vector $X$ collapses to zero as $\mu=0$ and $\Sigma\to 0$. Therefore, (\ref{logg}) leads to the well-known expression for the moment generating function in (\ref{momgen}):
$$
    g(\mu,\Sigma) = \re^{\lambda^{\rT}\mu + \frac{1}{2}\|\lambda\|_{\Sigma}^2}.
$$
Also note that Price's theorem admits a dynamic formulation for a Gaussian random process $X$. In this setting, the mean vector $\mu$, the covariance matrix $\Sigma$ and the generalized moment $g$ in (\ref{g}) acquire dependence on time, and the total time derivative of $g$ can be computed by combining the chain rule with (\ref{Price1}) and (\ref{Price2}) as
\begin{equation}
\label{dynamo}
    (\bE f(X))^{^\centerdot}
    =
    \dot{\mu}^{\rT} \bE \d_x f(X) +
    \frac{1}{2}
    \Bra
        \dot{\Sigma},
        \bE \d_x^2 f(X)
    \Ket.
\end{equation}
In particular, this relation becomes a directly verifiable identity for quadratic functions $f$. Indeed, in this case, $\d_x^2 f$ is a constant matrix and $\d_x f$ is an affine function of $x$, which allows $\bE \d_xf(X)$ to  be expressed in terms of $\mu$. More precisely, if
\begin{equation}
\label{quadro}
    f(x) := \beta^{\rT} x + \frac{1}{2} x^{\rT} R x,
    \qquad
    x \in \mR^n,
\end{equation}
where $\beta:= (\beta_j)_{1\< j \< n} \in \mR^n$ is a constant vector and $R:=  (r_{jk})_{1\< j,k\< n}\in \mS_n$ is a constant matrix, then  (\ref{dynamo}) takes the form
$$
    \left(\beta^{\rT}\mu + \frac{1}{2}\left(\mu^{\rT}R \mu + \bra R, \Sigma\ket\right)\right)^{\centerdot}
    =
    \dot{\mu}^{\rT} (\beta + R \mu) +
    \frac{1}{2}
    \Bra
        \dot{\Sigma},
        R
    \Ket.
$$

%%%%%%%%%%%%%%%%%%%%%%%%%%%%%%%%%%%%%%%%%%%%%%%%%%%%%%%%%%%%%%%%%%%%%%%%%%%%%%%%%%%%%%%%%%%%%%%%%%%
\section{QUANTUM MECHANICAL VERSION OF PRICE'S THEOREM}\label{sec:quantum}
%%%%%%%%%%%%%%%%%%%%%%%%%%%%%%%%%%%%%%%%%%%%%%%%%%%%%%%%%%%%%%%%%%%%%%%%%%%%%%%%%%%%%%%%%%%%%%%%%%%

Now, let $X:= (X_k)_{1\< k \< n}$ be a vector of $n$ quantum  variables, which are self-adjoint operators on a complex separable  Hilbert space $\cH$ representing real-valued physical quantities \cite{M_1998,S_1994}. Suppose the quantum variables $X_1, \ldots, X_n$ satisfy the Heisenberg CCRs (on a dense domain in $\cH$):
\begin{equation}
\label{theta}
    [X_j, X_k] =
    2i
    \theta_{jk}\cI,
    \qquad
    1\< j,k\< n.
\end{equation}
Here, $[\xi,\eta]:= \xi\eta-\eta \xi$ is the commutator of operators, $i:= \sqrt{-1}$ is the imaginary unit,
and $\Theta:= (\theta_{jk})_{1\<j, k\< n}$ is a real antisymmetric matrix of order $n$ (we denote the subspace of such matrices in $\mR^{n\x n}$ by $\mA_n$). Also, $\cI$ denotes the identity operator on the space $\cH$, which carries out an appropriate ampliation  of entries of the CCR matrix $\Theta$ to linear operators on $\cH$ and will be omitted for brevity. A vector-matrix form of the CCRs (\ref{theta}) is
\begin{equation}
\label{Theta}
    [X,X^{\rT}]
    :=
    \left([X_j,X_k]\right)_{1\< j,k\< n}
     =
    XX^{\rT} - (XX^{\rT})^{\rT}
    =
    2i
    \Theta,
\end{equation}
where $\Theta$ represents $\Theta \ox \cI$, with $\ox$ the tensor product, and the transpose $(\cdot)^{\rT}$ acts on matrices of operators as if their entries were scalars. In particular, the CCRs hold for the quantum mechanical position $q$ and momentum $p$ operators \cite{M_1998} given by
\begin{equation}
\label{qp}
    q
    =
    \frac{\fa + \fa^{\dagger}}{\sqrt{2}},
    \qquad
    p:= -i\d_q
    =
    \frac{\fa - \fa^{\dagger}}{i\sqrt{2}}
\end{equation}
on a dense domain in the Hilbert space of square integrable  complex-valued  functions on the real line of positions,
where $\d_q$ is the partial derivative with respect to the position variable $q$, and units are chosen so that the reduced Planck constant is $\hslash = 1$. Here,
$\fa = \frac{q + \d_q}{\sqrt{2}}$ and $\fa^{\dagger}= \frac{q - \d_q}{\sqrt{2}}$ are the annihilation and creation operators \cite[pp. 90--91]{S_1994}  satisfying the CCR $[\fa, \fa^{\dagger}] = 1$, with $(\cdot)^{\dagger}$ the operator adjoint. Accordingly, $[q,p] = -i[q,\d_q] = i$, and hence, the CCR matrix of the position and momentum operators in (\ref{qp}) is $\frac{1}{2}\bJ$ in the sense that
\begin{equation}
\label{bJ}
    \left[
        {\small\begin{bmatrix}
            q\\
            p
        \end{bmatrix}},
        {\small\begin{bmatrix}
            q &
            p
        \end{bmatrix}}
    \right]
    :=
        {\small\begin{bmatrix}
            [q,q] & [q,p]\\
            [p,q] & [p,p]
        \end{bmatrix}}
    =
    i \bJ,
    \quad
    \bJ
    :=
    {\small\begin{bmatrix}
        0& 1\\
        -1 & 0
    \end{bmatrix}},
\end{equation}
cf. (\ref{Theta}). Note that the matrix $\bJ$
spans the subspace $\mA_2$, and $-i\bJ$ is the second Pauli matrix \cite{S_1994}. The problem of evaluating a function $f: \mR^n \to \mR$ at  the noncommutative quantum variables $X_1, \ldots, X_n$ is nontrivial even if $f$ is a polynomial. The Weyl quantization \cite{F_1989} endows $f(X)$ with the following meaning:
\begin{equation}
\label{f}
    f(X):= \int_{\mR^n} F(\lambda) \re^{i\lambda^{\rT}X}\rd \lambda.
\end{equation}
In this definition, $\lambda^{\rT} X = \sum_{k=1}^n \lambda_k X_k$ is a self-adjoint operator on the underlying Hilbert space $\cH$,  which is a linear combination of the operators $X_1, \ldots, X_n$ with real coefficients $\lambda_1, \ldots, \lambda_n$
comprising the vector $\lambda := (\lambda_k)_{1\< k \< n}\in \mR^n$. Also, $F: \mR^n \to \mC$ is the Fourier transform of the real-valued function $f$:
\begin{equation}
\label{F}
    F(\lambda):= (2\pi)^{-n} \int_{\mR^n} f(x)\re^{-i\lambda^{\rT}x}\rd x = \overline{F(-\lambda)},
\end{equation}
with $\overline{(\cdot)}$ the complex conjugate. Note that the quantum mechanical definition of $f(X)$ replaces the complex number $\re^{i\lambda^{\rT}x}$ of unit modulus with the unitary Weyl operator $\re^{i\lambda^{\rT} X}$. Therefore, in the case when the Fourier transform $F$ is absolutely integrable, (\ref{f}) can be understood as a Bochner integral \cite{Y_1980} whose value is a bounded operator on the space $\cH$ with the norm bound
$$
    \|f(X)\|
    \<
    \int_{\mR^n}|F(\lambda)|\rd \lambda.
$$
The quantum variables $X_1, \ldots, X_n$ are said to be in a Gaussian quantum state \cite{P_2010}
if the corresponding quantum quasi-characteristic function \cite{CH_1971} of the vector $X$ has the following form:
\begin{equation}
\label{charfun}
    \bE \re^{i \lambda^{\rT}X}
    =
    \re^{i\lambda^{\rT}\mu- \frac{1}{2}\lambda^{\rT} S \lambda}
    =
    \re^{i\lambda^{\rT}\mu- \frac{1}{2}\|\lambda\|_{\Sigma}^2},
    \qquad
    \lambda \in \mR^n,
\end{equation}
where $\bE\xi:= \Tr (\rho\xi)$ denotes the expectation of a quantum variable $\xi$ over the density  operator $\rho$ which specifies the quantum state \cite{H_2001}. Also, $\mu:= \bE X$ is the mean vector of $X$ as before, and $\Sigma$  is the real part of a complex Hermitian matrix $S$ which is the quantum covariance matrix of the vector $X$:
\begin{equation}
\label{S}
    S:= \cov(X) = \Sigma + i\Theta,
    \qquad
    \Sigma:= \Re S.
\end{equation}
The imaginary part $\Im S = \Theta$ is the CCR matrix from (\ref{Theta}) which (in contrast to $\Sigma$) does not depend on the density operator $\rho$. In view of the generalized Heisenberger uncertainty principle \cite{H_2001}, the matrix $S$ in (\ref{S}) is positive semi-definite. Now, suppose the underlying density operator $\rho$ is varied so as to yield Gaussian states with different values of the parameters $\mu$ and $\Sigma$, while the  vector $X$ of quantum  variables remains  unchanged. This quantum mechanical setting admits the following analogue of Price's theorem.

%%%%%%%%%%%%%%%%%%%%%%%%%%%%%%%%%%%%%%%%%%%%%%%%%%%%%%%%%%%%%%%%%%%%%%%%%%%%%%%%%%%%%%%%%%%%%%%%%%%%%%%%%%%%%%%%%%%%%%%%%%%
\begin{thm}
\label{th:QPT}
Suppose the Fourier transform $F$ of the  function $f: \mR^n\to \mR$ in (\ref{F}) satisfies a weighted integrability condition
\begin{equation}
\label{square}
    \int_{\mR^n} |F(\lambda)|\left(1+|\lambda|^2\right) \rd \lambda <+\infty
\end{equation}
(which is stronger than the absolute integrability of $F$).
Also, let the vector $X$ of quantum variables, satisfying the CCRs (\ref{Theta}), be in a Gaussian quantum state with a positive definite quantum covariance matrix $S$ in (\ref{S}). Then the generalized moments of $X$, specified by the Weyl quantizations of $f$, the  gradient vector $\d_x f$ and   the Hessian matrix $\d_x^2f$ in the sense of (\ref{f}), are related by
\begin{align}
\label{dd}
    \d_{\mu} \bE f(X)
     & =
    \bE \d_xf(X),\\
\label{QPT}
    \d_{\Sigma} \bE f(X)
     & =
    \frac{1}{2}
    \bE \d_x^2f(X),\\
\label{mix}
    \d_{\mu}^2 \bE f(X)
    & =
    2\d_{\Sigma} \bE f(X).
\end{align}
\end{thm}
%%%%%%%%%%%%%%%%%%%%%%%%%%%%%%%%%%%%%%%%%%%%%%%%%%%%%%%%%%%%%%%%%%%%%%%%%%%%%%%%%%%%%%%%%%%%%%%%%%%%%%%%%%%%%%%%%%%%%%%%%%%
\begin{proof}
Similarly to the case of classical random variables discussed in Section~\ref{sec:classical}, the relation (\ref{dd}) follows from the identity
\begin{equation}
\label{z}
    \re^{i\lambda^{\rT}(X+z)} = \re^{i\lambda^{\rT}z} \re^{i\lambda^{\rT}X},
\end{equation}
which holds for all $\lambda,z \in \mR^n$, and the property that the mean vector  $\mu$ is a shift parameter of the Gaussian quantum state. Indeed, differentiation of (\ref{z}) with respect to the translation vector $z$ leads to
\begin{equation}
\label{trans}
    \d_{\mu}\bE \re^{i\lambda^{\rT}X} = \d_z\bE \re^{i\lambda^{\rT}(X+z)}\big|_{z=0} = i\lambda\bE \re^{i\lambda^{\rT}X},
\end{equation}
which can also be obtained directly from (\ref{charfun}). A combination of (\ref{f}) with (\ref{trans})  implies that
\begin{align}
\nonumber
    \d_{\mu} \bE f(X)
    & = \int_{\mR^n} F(\lambda)\d_{\mu}\bE\re^{i\lambda^{\rT}X}\rd \lambda\\
\nonumber
    & = i \int_{\mR^n} F(\lambda)\lambda \bE\re^{i\lambda^{\rT}X}\rd \lambda\\
\nonumber
    & = i \bE \int_{\mR^n} F(\lambda)\lambda \re^{i\lambda^{\rT}X}\rd \lambda = \bE \d_xf(X),
\end{align}
thus establishing (\ref{dd}). Here, use is also made of the Fourier transform
$$
    (2\pi)^{-n}\int_{\mR^n}\d_x f(x)\re^{-i\lambda^{\rT}x}\rd x = iF(\lambda) \lambda
$$
for the gradient $\d_x f$,
whereby the Weyl quantization of  $\d_xf(X)$ takes the form
$$
    \d_xf(X) = i \int_{\mR^n} F(\lambda)\lambda \re^{i\lambda^{\rT}X}\rd \lambda.
$$
The latter is a well-defined Bochner integral since the condition (\ref{square}) implies that $\int_{\mR^n}|F(\lambda)||\lambda|\rd \lambda < +\infty$.
We will now prove (\ref{QPT}).
Differentiation of the Gaussian cha\-rac\-teris\-tic function in (\ref{charfun}) with respect to the matrix $\Sigma$, satisfying $\Sigma \succ - i\Theta$ (such matrices $\Sigma$ form an open set in $\mS_n$), leads to
\begin{align}
\nonumber
    \d_{\Sigma}    \bE \re^{i \lambda^{\rT}X}
    & =
    \re^{i\lambda^{\rT}\mu}
    \d_{\Sigma}\re^{- \frac{1}{2}\|\lambda\|_{\Sigma}^2}\\
\label{dSigma}
    & =
    -\frac{1}{2}
    \re^{i\lambda^{\rT}\mu- \frac{1}{2}\|\lambda\|_{\Sigma}^2}
    \lambda\lambda^{\rT}
    =
    -\frac{1}{2}
    \lambda\lambda^{\rT}
    \bE \re^{i \lambda^{\rT}X}
\end{align}
for any fixed but otherwise arbitrary $\lambda \in \mR^n$.
Here, the Fre\-chet derivative $\d_{\Sigma} (\|\lambda\|_{\Sigma}^2) = \lambda \lambda^{\rT}$ corresponds to that  in the second of the relations  (\ref{F12}). Therefore, under the condition (\ref{square}), it follows from (\ref{f}) and (\ref{dSigma}) that
\begin{align}
\nonumber
    \d_{\Sigma} \bE f(X)
    & =
    \int_{\mR^n}
    F(\lambda) \d_{\Sigma} \bE \re^{i \lambda^{\rT}X}\rd \lambda\\
\label{dSigma1}
     & =
    -\frac{1}{2}
    \int_{\mR^n}
    F(\lambda) \lambda\lambda^{\rT}\bE \re^{i \lambda^{\rT}X}\rd \lambda.
\end{align}
The Fourier transform of the Hessian matrix $\d_x^2 f$ is representable as
\begin{equation}
\label{FHess}
    (2\pi)^{-n}\int_{\mR^n}\d_x^2 f(x)\re^{-i\lambda^{\rT}x}\rd x = -F(\lambda) \lambda\lambda^{\rT},
\end{equation}
and hence, the corresponding Weyl quantization of $\d_x^2f(X)$  is given by
\begin{equation}
\label{f''}
    \d_x^2f(X) = -\int_{\mR^n}F(\lambda) \lambda\lambda^{\rT}\re^{i\lambda^{\rT}X}\rd \lambda.
\end{equation}
The second integral in (\ref{dSigma1}) can be obtained by averaging that in (\ref{f''}), which leads to
$$
    \d_{\Sigma} \bE f(X)
     =
    -\frac{1}{2}
    \bE
    \int_{\mR^n}
    F(\lambda) \lambda\lambda^{\rT} \re^{i \lambda^{\rT}X}\rd \lambda\\
     =
    \frac{1}{2}
    \bE \d_x^2f(X),
$$
thus establishing (\ref{QPT}). Finally, by applying (\ref{dd}) twice and using (\ref{QPT}), it follows that
$$
    \d_{\mu}^2 \bE f(X) = \bE \d_x^2 f(X) = 2 \d_{\Sigma}\bE f(X),
$$
which proves the relation  (\ref{mix}) and completes the proof of the theorem.
\end{proof}
%%%%%%%%%%%%%%%%%%%%%%%%%%%%%%%%%%%%%%%%%%%%%%%%%%%%%%%%%%%%%%%%%%%%%%%%%%%%%%%%%%%%%%%%%%%%%%%%%%%%%%%%%%%%%%%%%%%%%%%%%%%%%%%

In establishing the quantum analogue of Price's theorem, we have essentially used the proof  \cite{B_1967,V_1999} of its original classical version based on Fourier transforms,  since the latter underlie the Weyl quantization. Note that quantum Price's theorem (\ref{QPT}) can, in principle,  be extended to the case, where the Fourier transform $F$ in (\ref{F}) is a generalized function \cite{V_2002}, with (\ref{f}) %(\ref{FHess})
being understood in an appropriate distributional sense.  This includes (but is not limited to) the class of polynomials $f$. For example,    the quadratic function $f$ in (\ref{quadro}) has the following Fourier transform:
\begin{equation}
\label{Fquadro}
    F(\lambda)
    = i\sum_{j=1}^n\beta_j \d_{\lambda_j}\delta(\lambda)-\frac{1}{2}\sum_{j,k=1}^n r_{jk}\d_{\lambda_j}\d_{\lambda_k}\delta(\lambda),
%    \\
%    &= i\beta^{\rT} \d_{\lambda}\delta(\lambda)-\frac{1}{2}\Tr \left(R\d_{\lambda}^2 \delta(\lambda)\right),
\end{equation}
where $\delta(\cdot)$ is the $n$-dimensional Dirac delta function.
Here, for any $n$-index $\alpha:= (\alpha_k)_{1\< k\< n} \in \mZ_+^n$ (with $\mZ_+$ the set of nonnegative integers), the value of the generalized function $\d_{\lambda}^{\alpha}\delta(\lambda)$ at an $|\alpha|$ times continuously differentiable function $\lambda \mapsto \varphi(\lambda)$ is $(-1)^{|\alpha|}\d_{\lambda}^{\alpha}\varphi(0)$, where the standard multiindex conventions $|\alpha|:= \alpha_1 + \ldots + \alpha_n$ and $\d_{\lambda}^{\alpha}:= \d_{\lambda_1}^{\alpha_1}\ldots \d_{\lambda_n}^{\alpha_n}$ are used. In the noncommutative case being considered, the mixed partial derivatives of $\re^{i\lambda^{\rT}X}$ with respect to the entries of the vector $\lambda\in \mR^n$ can be calculated by using the factorization
\begin{align}
\nonumber
    \re^{i \lambda^{\rT} X}
      & =
    \rprod_{k=1}^n
        \re^{i\lambda_k X_k-\frac{1}{2}\big[\sum_{j=1}^{k-1} i\lambda_j X_j,\ i\lambda_k X_k\big]}\\
\label{Baker}
        & =
    \re^{i\sum_{1\< j< k \< n}\theta_{jk} \lambda_j\lambda_k}
    \rprod_{k=1}^n
    \re^{i\lambda_k X_k}
    =
    \re^{\frac{i}{2}\lambda^{\rT}\wt{\Theta}\lambda}
    \rprod_{k=1}^n
    \re^{i\lambda_k X_k}.
\end{align}
The latter is obtained by repeated  application of the CCRs (\ref{theta}), the bilinearity of the commutator, and the  Baker-Campbell-Hausdorff formula $\re^{\xi+\eta} = \re^{\xi} \re^{\eta} \re^{-\frac{1}{2}[\xi,\eta]} $ for operators $\xi$ and $\eta$ which commute with their commutator
\cite[pp. 128--129]{GZ_2004}. Also, $ \rprod_{k=1}^n \xi_k := \xi_1\x \ldots  \x \xi_n $ denotes the ordered product of operators $\xi_1, \ldots, \xi_n$, and the matrix $\wt{\Theta}:= (\wt{\theta}_{jk})_{1\< j,k\< n} \in \mS_n$ is given by
$$    \wt{\theta}_{kj}
    :=
    \wt{\theta}_{jk}
    :=
    \theta_{jk},
    \qquad
    1\< j\< k\< n,
$$
thus inheriting zero diagonal entries from the CCR matrix $\Theta$. In particular, the CCRs of the position and momentum operators in (\ref{bJ}) lead  to
$$
    \wt{\Theta} = \frac{1}{2}{\begin{bmatrix}0 & 1\\ 1 & 0\end{bmatrix}}.
$$
From the product structure of the right-hand side of (\ref{Baker}), it follows  that
\begin{equation}
\label{diff}
    \d_{\lambda}^{\alpha}\re^{i\lambda^{\rT}X}\big|_{\lambda=0}
    =
    \alpha!
    \sum_{\gamma \in \mZ_+^n:\, \gamma \< \alpha}
    \frac{i^{|\gamma|}}{\gamma! (\alpha-\gamma)!}
    \d_{\lambda}^{\alpha-\gamma}\re^{\frac{i}{2}\lambda^{\rT}\wt{\Theta}\lambda}\big|_{\lambda=0}X^{\gamma}
\end{equation}
for any $\alpha \in\mZ_+^n$. Here,
the inequality $\gamma\< \alpha$ applies entry-wise, $\alpha!:= \alpha_1!\x \ldots \x \alpha_n!$, and  $X^{\alpha}:= \rprod_{k=1}^n X_k^{\alpha_k}$. Substitution of (\ref{Fquadro}) into (\ref{f}) and application of (\ref{diff}) indeed leads to the quadratic function $f(X) = \beta^{\rT}X + \frac{1}{2}X^{\rT}RX$ of the quantum variables.

%%%%%%%%%%%%%%%%%%%%%%%%%%%%%%%%%%%%%%%%%%%%%%%%%%%%%%%%%%%%%%%%%%%%%%%%%%%%%%%%%%%%%%%%%%%%%%%%%%%
\section{ILLUSTRATIVE EXAMPLE: QUADRATIC-EXPONENTIAL MOMENTS }\label{sec:examples}
%%%%%%%%%%%%%%%%%%%%%%%%%%%%%%%%%%%%%%%%%%%%%%%%%%%%%%%%%%%%%%%%%%%%%%%%%%%%%%%%%%%%%%%%%%%%%%%%%%%

Consider a quadratic-exponential moment
for the vector $X$ of quantum variables from Section \ref{sec:quantum} in a Gaussian state with the mean vector $\mu$ and quantum covariance matrix $S$ in (\ref{S}):
\begin{equation}
\label{QEM}
  g(\mu,\Sigma,\Pi) = \bE f(X),
  \quad
  f(x):= \re^{-\frac{1}{2}x^{\rT}\Pi x},
  \qquad
  x\in \mR^n.
\end{equation}
Here, the matrix $\Pi \in \mS_n$ plays the role of a parameter. In addition to being of interest to quantum probability in their own right (see, for example,
\cite[pp. 274--276]{M_1995}), such moments (with $\Pi\prec 0$) are employed in the risk-sensitive dissipativity analysis and filtering/control  design  \cite{DDJW_2006,J_2005,VP_2012b,YB_2009} for quantum stochastic systems. The asymptotic behaviour of  the quadratic-exponential moment $g$ in (\ref{QEM}) for small matrices $\Pi$ is described by
\begin{equation}
\label{gas}
    g = 1 - \frac{1}{2}\left(\|\mu\|_{\Pi}^2 + \bra \Sigma, \Pi\ket\right) + o(\Pi),
    \qquad
    \Pi \to 0.
\end{equation}
Since $\re^{\xi}\succcurlyeq \cI + \xi$ for any self-adjoint operator $\xi$ on the underlying Hilbert space $\cH$, the affine part of (\ref{gas}) provides  a lower bound:
$$
    g \> \bE\Big(\cI - \frac{1}{2}X^{\rT}\Pi X\Big) = 1 - \frac{1}{2}\left(\|\mu\|_{\Pi}^2 + \bra \Sigma, \Pi\ket\right).
$$
In the case $\Pi\succ 0$, the Fourier transform $F$ of the function $f$ in (\ref{QEM}) is a Gaussian PDF from (\ref{p}) with zero mean and covariance matrix $\Pi$:
\begin{align}
\nonumber
    F(\lambda)
    & = (2\pi)^{-n}\int_{\mR^n} \re^{-\frac{1}{2}\|x\|_{\Pi}^2-i\lambda^{\rT}x}\rd x\\
\label{FPi}
    & =
    \frac{(2\pi)^{-n/2}}{\sqrt{\det \Pi}}\re^{-\frac{1}{2}\|\lambda\|_{\Pi^{-1}}^2} = p_{0,\Pi}(\lambda),
    \qquad
    \lambda \in \mR^n,
\end{align}
and hence, $F$ satisfies the assumption (\ref{square}) of Theorem~\ref{th:QPT}. In the framework of the Weyl quantization for $f(X)$, the quadratic-exponential moment $g$ in (\ref{QEM}) can be computed by substituting (\ref{FPi}) into (\ref{f}) and using the Gaussian characteristic function (\ref{charfun}):
\begin{align}
\nonumber
    \bE \re^{-\frac{1}{2}X^{\rT}\Pi X}
    & =
    \int_{\mR^n}
    p_{0,\Pi}(\lambda)
        \bE \re^{i\lambda^{\rT}X}
    \rd \lambda\\
%\nonumber
%    & =
%    \int_{\mR^n}
%    p_{0,\Pi}(\lambda)
%    \re^{i\lambda^{\rT}\mu- \frac{1}{2}\|\lambda\|_{\Sigma}^2}
%    \rd \lambda\\
\nonumber
    & =
    \frac{(2\pi)^{-n/2}}{\sqrt{\det \Pi}}
    \int_{\mR^n}
    \re^{i\lambda^{\rT}\mu-\frac{1}{2}\|\lambda\|_{\Pi^{-1}+\Sigma}^2}\,
    \rd \lambda\\
\label{ggg}
     & =
    \frac{\re^{-\frac{1}{2}\|\mu\|_{\Psi}^2}    }{\sqrt{\det(I_n + \Sigma \Pi)}},
\end{align}
where $I_n$ denotes the identity matrix of order $n$, and $\Psi\in \mS_n$ is an auxiliary matrix given by
\begin{equation}
\label{Psi}
    \Psi:= (\Pi^{-1}+\Sigma)^{-1} = \Pi (I_n+ \Sigma\Pi)^{-1}.% = (I_n+ \Pi\Sigma)^{-1}\Pi.
\end{equation}
The closed-form representation of the moment $g$ in (\ref{ggg}) allows the relation (\ref{mix}) of Theorem~\ref{th:QPT} to be verified directly  as
$$
    \d_{\mu}^2 g = g (\Psi \mu \mu^{\rT} \Psi - \Psi) = 2 \d_{\Sigma}g,
$$
because a combination of the identities (\ref{F12}) with (\ref{Psi}) implies that
\begin{align*}
    \d_{\Sigma} \ln\det (I_n + \Sigma \Pi) & = \Psi,\\
    \d_{\Sigma} (\|\mu\|_{\Psi}^2) & = -\Psi \mu\mu^{\rT}\Psi.
\end{align*}
Note that the calculations in (\ref{ggg})  substantially rely on the Weyl quantization of the quadratic-exponential function $f$ in (\ref{QEM}) under the assumption that $\Pi\succ 0$. The result would be different if the power series $\sum_{k=0}^{+\infty}\frac{1}{k!}(-\frac{1}{2}X^{\rT}\Pi X)^k$ were used instead.

\section{CONCLUSION}\label{sec:conc}

We have revisited  Price's classical theorem for generalized moments of jointly Gaussian random variables  by using a unified apparatus of the  Frechet differentiation with respect to covariance matrices. By combining the quantum quasi-characteristic functions with Fourier transforms, we have shown that similar integro-differential identities hold for expectations of Weyl quantization integrals evaluated at quantum variables satisfying Heisenberg CCRs in Gaussian states. The quantum mechanical version of Price's theorem involves the Frechet derivative of the generalized moment of such variables with respect to the real part of their quantum covariance matrix. We have considered an illustrative example of the quadratic-exponential moments. % and discussed their connection with an analogue of the Cameron-Martin-Girsanov change of measure for Gaussian quantum states.
The techniques, which have been used in this paper,  are applicable to computing nonlinear performance criteria for linear quantum stochastic systems, such as those in risk-sensitive quantum control problems.

%\section*{ACKNOWLEDGMENT}
%
%The author thanks Professor Ian R. Petersen for useful discussions.


\begin{thebibliography}{99}
%\bibitem{A_1982}
%B.D.O.Anderson, Reverse time diffusion equation models, \emph{ Stoch. Process. Appl.}, vol.  12, no. 3, pp. 313--326.
%==============================================================================
%\bibitem{B_2010}
%V.P.Belavkin, Noncommutative dynamics and generalized master equations, \emph{Math. Notes}, vol. 87, no. 5, 2010, pp. 636--653.
%==============================================================================
%\bibitem{BH_1998}
%D.S.Bernstein, and W.M.Haddad,
%LQG control with an
%$H^{\infty}$ performance bound: a Riccati equation approach,
%\textit{IEEE Trans.
%Automat. Contr.}, vol. 34, no. 3, 1989, pp. 293--305.

%%==============================================================================
%\bibitem{B_1986}
%W.M.Boothby,
%\textit{An Introduction to Differentiable Manifolds and Riemannian Geometry}, 2nd Ed.,
%Academic Press, London, 1986.

%%==============================================================================
\bibitem{B_1967}
J.L.Brown, A Generalized form of Price's theorem and its converse,
\emph{IEEE Trans. Inform. Theory}, vol. 13, no. 1, 1967, pp. 27--30.
%%==============================================================================
\bibitem{CH_1971}
C.D.Cushen, and R.L.Hudson, A quantum-mechanical central limit theorem,
\emph{J. Appl. Prob.}, vol. 8, no. 3, 1971, pp. 454--469.
%%==============================================================================
\bibitem{DDJW_2006}
C.D'Helon, A.C.Doherty, M.R.James, and S.D.Wilson,
Quantum risk-sensitive control,
Proc. 45th IEEE CDC,
San Diego, CA, USA, December 13--15, 2006, pp. 3132--3137.
%%==============================================================================
%\bibitem{EB_2005}
%S.C.Edwards, and V.P.Belavkin,
%Optimal quantum filtering and
%quantum feedback control,
%arXiv:quant-ph/0506018v2, August 1,  2005.
%%==============================================================================
\bibitem{E_1998}
L.C.Evans,
\textit{Partial Differential Equations},
American Mathematical Society, Providence, 1998.
%%==============================================================================
\bibitem{F_1989}
G.B.Folland, \emph{Harmonic Analysis in Phase Space}, Princeton University Press, Princeton, 1989.
%%%==============================================================================
%\bibitem{GHP_2009}
%P.Gibilisco, F.Hiai, and D.Petz,
%Quantum covariance, quantum Fisher information, and the uncertainty principle,
%\textit{IEEE Trans.
%Inform. Theory.}, vol. 55, no. 1, 2009, pp. 439--443.
%%==============================================================================

\bibitem{GZ_2004}
C.W.Gardiner, and P.Zoller,
\textit{Quantum Noise}.
Springer, Berlin, 2004.
%%==============================================================================
%\bibitem{G_2006}
%M.de~Gosson,
%\textit{Symplectic Geometry and Quantum Mechanics},
%Birk\-h\"{a}user, Basel, 2006.
%%==============================================================================
%\bibitem{G_1977}
%H.W.Guggenheimer,
%\textit{Differential Geometry},
%Dover, New York, 1977.
%%==============================================================================
\bibitem{H_2001}
A.S.Holevo, \textit{Statistical Structure of Quantum Theory}, Springer, Berlin, 2001.
%%==============================================================================
\bibitem{HJ_2007}
R.A.Horn, and C.R.Johnson,
\textit{Matrix Analysis},
Cambridge
University Press, New York, 2007.
%%==============================================================================
%\bibitem{HM_1994}
%U.Helmke, and J.B.Moore,
%\textit{Optimization and Dynamical Systems},
%Springer, London, 1994.

%%==============================================================================
\bibitem{J_2005}
M.R.James, A quantum Langevin formulation of risk-sensitive optimal control,
\emph{J. Opt. B}, vol. 7, 2005, pp. S198--S207.
%%==============================================================================
%\bibitem{JNP_2008}
%M.R.James, H.I.Nurdin, and I.R.Petersen,
%$H^{\infty}$ control of
%linear quantum stochastic systems,
%\textit{IEEE Trans.
%Automat. Contr.}, vol. 53, no. 8, 2008, pp. 1787--1803.
%
%%==============================================================================
%\bibitem{KS_1972}
%H.Kwakernaak, and R.Sivan,
%\textit{Linear Optimal Control Systems},
%Wiley, New York, 1972.
%%==============================================================================
%\bibitem{LS_2001}
%R.S.Liptser, and A.N.Shiryaev,
%\textit{Statistics of Random Processes: Applications}, Springer, Berlin, 2001.
%%==============================================================================

%\bibitem{M_1988}
%J.R.Magnus,
%\textit{Linear Structures},
%Oxford University Press, New York, 1988.
%
%%==============================================================================
\bibitem{M_1964}
E.McMahon, An extension of Price's theorem, \emph{IEEE Trans. Inform. Theory}, vol. 10, no. 2, 1964, p. 168.
%%==============================================================================
\bibitem{M_1998}
E.Merzbacher,
\textit{Quantum Mechanics}, 3rd Ed.,
Wiley, New York, 1998.
%%==============================================================================
\bibitem{M_1995}
P.-A.Meyer, \emph{Quantum Probability for Probabilists}, Springer, Berlin, 1995. %pp. 274-276
%%==============================================================================

%
%\bibitem{NJP_2009}
%H.I.Nurdin, M.R.James, and I.R.Petersen,
%Coherent quantum LQG
%control,
%\textit{Automatica}, vol.  45, 2009, pp. 1837--1846.
%
%%==============================================================================
\bibitem{P_1965}
A.Papoulis, Comments on `An extension of Price's theorem' by McMahon, E.L.,
\emph{IEEE Trans. Inform. Theory},
vol. 11, no. 1,  1965, p. 154.
%%==============================================================================

\bibitem{P_1992}
K.R.Parthasarathy,
\textit{An Introduction to Quantum Stochastic Calculus},
Birk\-h\"{a}user, Basel, 1992.
%%==============================================================================
\bibitem{P_2010}
K.R.Parthasarathy,
What is a Gaussian state?
\textit{Commun. Stoch. Anal.}, vol. 4, no. 2, 2010, pp. 143--160.

%%==============================================================================
\bibitem{P_2010}
I.R.Petersen,
Quantum linear systems theory,
Proc. 19th Int. Symp. Math. Theor. Networks Syst., Budapest, Hungary, July 5--9, 2010, pp.  2173--2184.

%%==============================================================================
%\bibitem{PBGM_1962}
%L.S.Pontryagin, V.G.Boltyanskii,
%R.V.Gamkrelidze, and E.F. Mishchenko,
%\emph{The Mathematical Theory of Optimal Processes},
%Wiley, New York, 1962.
%==============================================================================
\bibitem{P_1958}
R.Price, A useful theorem for nonlinear devices having Gaussian inputs,
\emph{IRE Trans. Inform. Theory}, vol. 4, no. 2, 1958, pp. 69--72.
%==============================================================================
\bibitem{RA_1978}
E.B.Rockower,  and N.B.Abraham,
Calculating generating functions from characteristic
functions, with application to quantum optics,
\emph{J. Phys. A: Math. Gen.}, vol. 11, no. 10, 1978, pp. 1879--1884.
%==============================================================================
%\bibitem{SIS_2004}
%A.Serafini, F.Illuminati, and S.De Siena,
%Symplectic invariants, entropic measures and correlation of Gaussian states,
%\textit{J. Phys. B: At. Mol. Opt. Phys.}, vol. 37, 2004, pp. L21--L28.
%==============================================================================
%

\bibitem{S_1994}
J.J.Sakurai,
\textit{Modern Quantum Mechanics},
 Addison-Wesley, Reading, Mass., 1994.
%==============================================================================

%\bibitem{SP_2009}
%A.J.Shaiju, and I.R.Petersen,
%On the physical realizability of
%general linear quantum stochastic differential equations with
%complex coefficients,
%Proc. Joint 48th IEEE Conf. Decision Control \&
%28th Chinese Control Conf.,
%Shanghai, P.R. China, December 16--18, 2009, pp. 1422--1427.


%%==============================================================================
%
%\bibitem{S_2000}
%R.Simon,
%Peres-Horodecki separability criterion for continuous variable systems,
%\textit{Phys. Rev. Lett.},
%vol. 84, no. 12, 2000, pp. 2726--2729.

%
%%==============================================================================
%%
%\bibitem{SIG_1998}
%R.E.Skelton, T.Iwasaki, and K.M.Grigoriadis,
%\textit{A Unified Algebraic Approach to Linear Control Design},
%Taylor \& Francis, London, 1998.
%%%==============================================================================

\bibitem{S_2008}
D.W.Stroock,
Partial differential equations for probabilists,
Cambridge University Press, Cambridge, 2008.
%%%==============================================================================
%\bibitem{SW_1997}
%H.J.Sussmann, and J.C.Willems,
%300 years of optimal control: from the brachystochrone to the maximum principle,
%\textit{Control Systems}, vol. 17, no. 3, 1997, pp. 32--44.

\bibitem{V_1999}
A.van den Bos, Nonlinear statistical signal processing: useful theorems and their application, Proc. IEEE-EURASIP Workshop on Nonlinear Signal and Image Processing (NSIP'99), Antalya, Turkey, June 20--23, 1999, pp. 603--606.
%{\tt http://www.eurasip.org/Proceedings/Ext/NSIP99/Nsip99/ papers/130.pdf}.




%
%==============================================================================
\bibitem{V_1971}
V.S.Vladimirov,
\textit{Equations of Mathematical Physics},
M.Dekker,
New York, 1971.
%==============================================================================
\bibitem{V_2002}
V.S.Vladimirov,
\textit{Methods of the Theory of Generalized Functions},
Taylor \& Francis, London, 2002.

%%%==============================================================================
%\bibitem{VP_2010a}
%I.G.Vladimirov, and I.R.Petersen,
%Minimum relative entropy state transitions in linear stochastic
%systems: the continuous time case,
%Proc. 19th Int. Symp. Math. Theor. Networks Syst., Budapest, Hungary, July 5--9,  2010, pp.  51--58.
%%%==============================================================================
%\bibitem{VP_2010b}
%I.G.Vladimirov, and I.R.Petersen,
%Hardy-Schatten norms of systems, output energy cumulants and linear quadro-quartic  Gaussian control,
%Proc. 19th Int. Symp. Math. Theor. Networks Syst., Budapest, Hungary, July 5--9,  2010, pp.  2383--2390.
%%==============================================================================
%\bibitem{VP_2011a}
%I.G.Vladimirov, and I.R.Petersen,
%A quasi-separation principle and Newton-like scheme for coherent quantum LQG control, 	
%18th IFAC World Congress, Milan, Italy, 28 August--2 September, 2011, pp. 4721--4727.
%(preprint:  arXiv:1010.3125v2 [quant-ph], 15 April, 2011).
%%==============================================================================
%\bibitem{VP_2011b}
%I.G.Vladimirov, and I.R.Petersen,
%A dynamic programming approach to finite-horizon coherent quantum LQG control, 	
%Proc. Australian Control Conference, Melbourne, 10--11 November, 2011, pp. 357--362.
%(preprint:  	arXiv:1105.1574v1 [quant-ph], 9 May, 2011).
%%==============================================================================
\bibitem{VP_2012a}
I.G.Vladimirov, and I.R.Petersen,
Gaussian stochastic  linearization for open quantum systems
using quadratic approximation of Hamiltonians, 	Proc. MTNS 2012, Melbourne, Victoria, 9--13 July 2012,
%{\sf\scriptsize http://mtns2012.com.au/academic-program.php},
(preprint: arXiv:1202.0946v1 [quant-ph], 5 February 2012).
%%==============================================================================
\bibitem{VP_2012b}
I.G.Vladimirov, and I.R.Petersen,
Risk-sensitive dissipativity of linear quantum stochastic systems
under Lur'e type perturbations of Hamiltonians, 	Proc. AUCC 2012, Sydney, Australia, 15--16 November 2012, pp. 247--252 (preprint:  	
arXiv:1205.3566v1 [quant-ph], 16 May 2012).

%%==============================================================================
%\bibitem{W_1936}
%J.Williamson,
%On the algebraic problem concerning the normal forms of linear dynamical systems,
%\textit{Am. J. Math.}, vol. 58, no. 1, 1936, pp. 141--163.
%%==============================================================================
%\bibitem{W_1937}
%J.Williamson,
%On the normal forms of linear canonical transformations in dynamics,
%\textit{Am. J. Math.}, vol. 59, no. 3, 1937, pp. 599--617.
%%==============================================================================
\bibitem{YB_2009}
N.Yamamoto, and L.Bouten,
Quantum risk-sensitive estimation and robustness,
\emph{IEEE Trans. Automat. Contr.}, vol. 54, no. 1, 2009, pp. 92--107.
%%==============================================================================
\bibitem{Y_1980}
K.Yosida, \emph{Functional Analysis}, 6th Ed., Springer, Berlin, 1980.

\end{thebibliography}
\end{document}